\documentclass[conference]{IEEEtran}
\IEEEoverridecommandlockouts

\usepackage{cite}
\usepackage{amsmath,amssymb,amsfonts}
\usepackage{algorithmic}
\usepackage{graphicx}
\usepackage{textcomp}
\usepackage{xcolor}
\usepackage{amsthm}
\newtheoremstyle{asccstyle}{3pt}{3pt}{\normalfont}{\parindent}{\itshape}{.}{0.5em}{}
\theoremstyle{asccstyle}
\newtheorem{theorem}{Theorem}
\newtheorem{lemma}{Lemma}
\newtheorem{definition}{Definition}
\newtheorem{remark}{Remark}
\def\BibTeX{{\rm B\kern-.05em{\sc i\kern-.025em b}\kern-.08em
    T\kern-.1667em\lower.7ex\hbox{E}\kern-.125emX}}
\begin{document}

\title{Data-Driven Distributed Stability Certification for Power Systems via Input-State Trajectories}

\author{
    \IEEEauthorblockN{Xiaohui Zhang, Liaoyuan Yang, Peng Yang, \IEEEmembership{Member, IEEE}}
    \thanks{This work was supported by the National Natural Science Foundation of China (No. 52407137) and the Young Talent Fund of Association for Science and Technology in Shaanxi, China (No. 20250431). (Corresponding author: Peng Yang)}
    }

\maketitle

\begin{abstract}
This article proposes a data-driven framework to verify the distributed conditions that guarantee the system-wide stability for interconnected power systems. To guarantee system-wide stability, the dynamics of each bus are required to satisfy an output differential passivity (ODP) condition with a sufficient index. These ODP indices uniformly quantify the impacts on the system-wide stability of individual bus dynamics and the coupling strength from the power network. To obtain these indices without explicit physical models, we derive a data-driven linear matrix inequality (LMI) criterion based exclusively on measured input-state trajectories. Furthermore, extracting the optimal ODP index is formulated as a convex semi-definite programming (SDP) problem. Simulations verify the effectiveness of the proposed method under both single-device offline evaluation and system-wide online certification scenarios.
\end{abstract}


\section{Introduction}
\label{sec:introduction}

Modern power systems are undergoing a profound transition with the continuous advancement of renewable energy technologies \cite{hatziargyriou2020stability}. Specifically, the proliferation of massive distributed energy resources is shifting the traditional centralized and homogeneous generation paradigm towards a decentralized and highly heterogeneous mode, introducing complex and nonlinear dynamic behaviors into the grid \cite{guerrero2012advanced}. Such a rapid expansion in network dimensionality, coupled with the deep heterogeneity of underlying dynamic features, poses fundamental challenges to traditional system-wide stability analysis \cite{milano2018foundations}. Traditional analysis paradigms typically treat the entire network as a whole, relying on the formulation of high-dimensional differential-algebraic equations that cover the entire grid. These methods employ demanding time-domain numerical simulations \cite{stott2005power}, eigenvalue calculations based on the system-wide state matrix \cite{sastry2003hierarchical}, or efforts to construct a global energy function for the system \cite{chang1995direct}. However, in such complex and heterogeneous power systems, these traditional centralized methods often become intractable or even fail to solve effectively due to severe computational burdens, potential communication failures, and increasingly strict data privacy concerns \cite{liu2022stability},\cite{yang2019toward}. Consequently, there is an urgent need to develop stability analytics capable of accommodating these heterogeneous and nonlinear bus dynamics.

In this context, distributed stability analytics have become a more promising solution. So far, distributed stability conditions have been established mainly by exploiting known system structures or network decomposition techniques, such as reconstructing local Jacobian matrices \cite{song2017distributed}, utilizing the passivity of interconnected systems \cite{yang2019distributed}, and deriving distributed parametric conditions based on subsystem models \cite{nandanoori2020distributed}. However, these conventional methods exhibit inherent limitations in practical deployment: they are fundamentally model-based and strictly rely on explicit differential-algebraic equations. In modern power grids, obtaining such precise analytical models is often highly challenging. On the one hand, due to commercial confidentiality, inverters frequently behave as "black boxes" with unknown internal control mechanisms \cite{men2023artificial}. On the other hand, unmodeled dynamics introduced by complex active loads further degrade the accuracy of theoretical modeling \cite{ruiz2019data}. To overcome this model dependency, data-driven frameworks leveraging high-fidelity measurement trajectories offer an effective alternative.

Rooted in behavioral systems theory, Willems' fundamental lemma shows that under persistently exciting conditions, a system's trajectories can be exactly characterized by a Hankel matrix constructed from a single historical input-state trajectory \cite{willems2005note}. This theory lays a solid mathematical foundation for bypassing explicit state-space equations and synthesizing analytical conditions directly from measurement data \cite{de2019formulas}. Following this theoretical line, pioneering frameworks \cite{koch2020verifying, 9551767} have successfully achieved rigorous verification of standard passivity for unknown systems using input-state measurement trajectories. However, standard passivity often exhibits inherent conservatism when dealing with complex nonlinear dynamics, making it difficult to accurately characterize transient incremental behaviors under small disturbances.

Meanwhile, data-driven methods have also made significant progress in power systems, being progressively applied to scenarios such as model-free adaptive control of interlinked microgrids \cite{zhang2015data} and coordinated optimization of inverters \cite{xu2021data}, among others. However, existing studies predominantly focus on localized controller design or steady-state optimization. For interconnected systems comprising massive heterogeneous devices, utilizing data-driven methods to achieve rigorous system-wide distributed small-signal stability certification remains a critical open challenge.

To address this challenge and overcome the conservatism of passivity, this paper proposes a data-driven distributed output differential passivity (ODP) evaluation and stability certification framework for interconnected power systems. The main contributions are summarized as follows:

\begin{itemize}
    \item Utilizing input-state trajectory data, the ODP property of nonlinear buses is directly cast into a linear matrix inequality (LMI) for solution, without the need for explicit internal dynamic models of the devices.
    \item Combining the network topology information of the interconnected system with the data-driven ODP indices derived from individual buses, a system-wide distributed small-signal stability criterion is established.
\end{itemize}

The effectiveness of the proposed method is verified through simulations on a single-device offline case and a 3-bus system online case.

\section{PRELIMINARIES AND FORMULATIONS}

This section briefly reviews the fundamental concepts of continuous-time ODP and introduces the distributed stability criterion for interconnected power systems.

\subsection{Output Differential Passivity}

The concept of ODP, introduced in \cite{yang2019distributed}, is established upon the continuous-time nonlinear system described by:
\begin{equation}\label{eq:nonlinear-bus}
\begin{cases}
\dot{x} = f(x, u) \\
y = h(x)
\end{cases},
\end{equation}
where $x \in \mathbb{R}^n$, $u \in \mathbb{R}^m$ and $y \in \mathbb{R}^m$. For the simplicity of notations, we assume the zero input-state-output equilibrium triplet $(u^*, x^*, y^*) = (0, 0, 0)$. Note, however, all the following definitions are still valid for nonzero input-state-output triplet $(u^*, x^*, y^*)$ by simply replacing $u$, $x$, and $y$ by $u-u^*$, $x-x^*$, and $y-y^*$.

\begin{definition}[ODP\cite{yang2019distributed}]
System (1) is said to be output-differential passive (OD-passive) if there exists a continuously differentiable function $S(x)$, called \textit{storage function}, such that $x = 0$ is a local minimum of $S(x)$ and
\begin{equation}
    \dot{S} \le u^\top \dot{y} - \sigma y^\top \dot{y} = (u - \sigma y)^\top \dot{y},
\end{equation}
for some scalar $\sigma \in \mathbb{R}$, denoted as ODP($\sigma$).
\end{definition}

The scalar $\sigma$ serves to quantify the excess or shortage of the system's output differential passivity, which is formally termed the ODP index. Specifically, a positive scalar $\sigma > 0$ signifies a surplus of ODP, whereas a negative scalar $\sigma < 0$ indicates an ODP deficit.

\subsection{Interconnected Power Systems and Distributed Stability}

Consider a continuous-time interconnected power system comprising $N$ heterogeneous buses, whose system model is detailed in \cite{yang2019distributed}. For each bus $i$, the input vector is defined as the power injection $u_i := [-P_i, -Q_i/V_i]^\top$ and the output vector as the bus voltage $y_i := [\theta_i, V_i]^\top$. 

To establish the analytical framework for system-wide stability verification, the fundamental bus-level physical properties and associated network-level criterion derived therein are introduced as follows:

\begin{lemma}[Distributed Stability Criterion \cite{yang2019distributed}]
The equilibrium $x^*$ of the continuous-time interconnected power system is Lyapunov stable if, for each bus $i \in \{1, 2, \dots, N\}$, the local dynamics satisfies the ODP($\sigma_i$) property with respect to the incremental input $\tilde{u}_i := u_i - u_i^*$ and output $\tilde{y}_i := y_i - y_i^*$, and all local ODP indices satisfy:
\begin{equation}
    \sigma_i > \sigma_{net},
\end{equation}
where the network-dependent threshold $\sigma_{net} \in \mathbb{R}$ is determined by the smallest eigenvalue associated with the network's energy function $W_B(y^*)$ and the loss function $\phi(y^*)$ at the equilibrium $y^*$, given by \cite{yang2019distributed}:
\begin{equation}
    \sigma_{net} := -\lambda_{\min} \bigg( \nabla^2 W_B(y^*) + \frac{1}{2} ( \nabla \phi(y^*) + \nabla \phi(y^*)^\top ) \bigg).
\end{equation}

Here, $\nabla^2 W_B(y^*)$ is the Hessian matrix corresponding to the conservative lossless power network, while $\nabla \phi(y^*)$ is the gradient of the loss function associated with line conductances. 

In inequality (3), the index $\sigma_i$ exclusively quantifies the output differential passivity margin provided by a local device, essentially reflecting its internal dynamic damping capability. Meanwhile, $\sigma_{net}$ represents the macroscopic coupling requirement of the grid. 

Physically, the condition (3) establishes a rigorous analytical decoupling between the local equipment and the macroscopic global grid. It implies that the system-wide stability can be ensured in a distributed manner, provided that the internal damping of each local device $\sigma_i$ is sufficient to compensate for the network's coupling effect $\sigma_{net}$.

\end{lemma}

In the remainder of this article, we propose a data-driven method relying exclusively on measured input-state trajectories to compute the ODP index $\sigma_i$ without knowing the bus dynamics $f$. We assume that the output function $h$ is known or, instead, the output trajectories are directly measurable.

\section{LTI System Formulations}

This section formulates the continuous-time and discrete-time linear time-invariant (LTI) systems, establishes their respective ODP conditions, and analyzes the theoretical approximation error introduced by the sampling process.

\subsection{ODP for Continuous-Time LTI System}
Small-signal stability concerns the local dynamic behavior of a system subject to minor perturbations around its steady-state equilibrium. Based on this physical premise, the nonlinear dynamics in (1) can be reasonably approximated via local linearization, yielding the following continuous-time LTI system:
\begin{equation}
\begin{cases}
\dot{x} = \bar{A}x + \bar{B} u \\
y = Cx
\end{cases},
\end{equation}
where $x \in \mathbb{R}^n$, $u \in \mathbb{R}^m$, and $y \in \mathbb{R}^m$ denote the state, input, and output vectors, respectively. The system matrices are given by $\bar{A} \in \mathbb{R}^{n \times n}$, $\bar{B} \in \mathbb{R}^{n \times m}$, and $C \in \mathbb{R}^{m \times n}$.

In practical multi-input multi-output (MIMO) systems, the ODP margin often exhibits an uneven distribution across different dimensions. To accurately characterize this multidimensional property, we extend the $\text{ODP}(\sigma)$ condition introduced in Section II to a generalized matrix form, which is formally defined as follows:

\begin{definition}[Generalized Matrix ODP] 
A continuous-time system is said to satisfy the $\text{ODP}(\Sigma)$ property if there exists a continuously differentiable storage function $S(x)$ and a symmetric matrix $\Sigma \in \mathbb{R}^{m \times m}$ such that the following dissipation inequality holds:
\begin{equation}
    \dot{S} \le u^\top \dot{y} - y^\top \Sigma \dot{y}. \label{eq:matrix_odp_def}
\end{equation}
\end{definition}
Meanwhile, to ensure compatibility with the system-wide stability verification in Lemma 1 and provide a unified metric, the equivalent scalar ODP index is extracted via the relationship $\Sigma \succeq \sigma I$. Specifically, this scalar is determined by the minimum eigenvalue of the symmetric matrix $\Sigma$ as:
\begin{equation}
    \sigma = \lambda_{\min}(\Sigma). \label{eq:sigma_extraction}
\end{equation}

Based on this generalized matrix formulation, the following lemma establishes the ODP criterion for continuous-time linear time-invariant (LTI) systems.

\begin{lemma}
The continuous-time LTI system (5) satisfies the ODP($\Sigma$) property if and only if one can find a symmetric positive-definite matrix $P \succ 0$ satisfying the following linear matrix inequality:
\begin{equation}
\bar{M}_1 - \frac{1}{2}(\bar{M}_2 + \bar{M}_2^\top) \preceq 0,
\end{equation}
where the corresponding block matrices are defined as:
\begin{equation}
\bar{M}_1 := \begin{bmatrix} \bar{A}^\top P + P\bar{A} & P\bar{B} \\ * & 0 \end{bmatrix} \\,
\end{equation}
\begin{equation}
\bar{M}_2 := \begin{bmatrix} -C^\top \Sigma C\bar{A} & -C^\top \Sigma C\bar{B} \\ C\bar{A} & C\bar{B} \end{bmatrix}.
\end{equation}
\end{lemma}
\textit{Proof.} Consider the storage function $S(x) = x^\top Px$ and define the augmented vector $z = [x^\top, u^\top]^\top$. Evaluating the time derivative yields:
\begin{align*}
\dot{S}(x) &= x^\top (\bar{A}^\top P + P\bar{A}) x + 2x^\top P \bar{B} u = z^\top \bar{M}_1 z
\end{align*}

Substituting $y = Cx$ and $\dot{y} = C\bar{A}x + C\bar{B}u$ into the right-hand side of the ODP dissipation inequality gives:
\begin{align*}
&u^\top \dot{y} - y^\top \Sigma \dot{y} \\
&= u^\top (C\bar{A}x + C\bar{B}u) - x^\top C^\top \Sigma (C\bar{A}x + C\bar{B}u) \\
&= z^\top \begin{bmatrix} -C^\top \Sigma C\bar{A} & -C^\top \Sigma C\bar{B} \\ C\bar{A} & C\bar{B} \end{bmatrix} z = z^\top \bar{M}_2 z
\end{align*}

The ODP condition requires $\dot{S}(x) \le u^\top \dot{y} - y^\top \Sigma \dot{y}$, which equates to $z^\top (\bar{M}_1 - \bar{M}_2) z \le 0$ for any arbitrary $z$. Enforcing symmetry for the central matrix directly yields the LMI condition $\bar{M}_1 - \frac{1}{2}(\bar{M}_2 + \bar{M}_2^\top) \leq 0$. \hfill $\square$

\subsection{ODP for Discrete-Time LTI System}
Data-driven verification frameworks inherently operate on sampled measurement sequences rather than continuous physical trajectories. Aligning with this data acquisition reality, the continuous-time system (5) is discretized, yielding the following discrete-time linear time-invariant (dt-LTI) model:
\begin{equation}
\begin{cases}
x_{k+1} = Ax_k + Bu_k \\
y_k = Cx_k
\end{cases},
\end{equation}
where $x_k \in \mathbb{R}^n$, $u_k \in \mathbb{R}^m$, and $y_k \in \mathbb{R}^m$ denote the state, input, and output vectors at the $k$-th sampling instant, respectively. The system matrices $(A, B)$ are related to the continuous-time matrices $(\bar{A}, \bar{B})$ via the step-invariant transformation with sampling interval $h > 0$:
\begin{equation}
A = e^{\bar{A}h}, \quad B = \int_{0}^{h} e^{\bar{A}(h-\tau)} \bar{B} d\tau,
\end{equation}
where $\exp(\bar{A}h)$ is the matrix exponential and is calculated as
\begin{equation}
    \exp(\bar{A}h) = \sum_{k=0}^{\infty} \frac{1}{k!} (\bar{A}h)^k.
\end{equation}

To preserve the physical characteristics of Output Differential Passivity in the discrete domain, a finite difference approximation is employed. By defining $\Delta y_k := y_{k+1} - y_k$ as the discrete-time counterpart of the output derivative, the ODP criterion is formulated as follows:

\begin{definition}
The dt-LTI system (11) is said to satisfy ODP($\Sigma$) if there exists a positive-definite matrix $P \succ 0$ such that the storage function $S(x_k) = x_k^\top P x_k$ satisfies the following dissipation inequality:
\begin{equation}
S(x_{k+1}) - S(x_k) \le u_k^\top \Delta y_k - y_k^\top \Sigma \Delta y_k.
\end{equation}
\end{definition}
Based on this quadratic storage function, the dissipation inequality in (14) is recast into a tractable linear matrix inequality (LMI) problem, as established in the following theorem.

\begin{theorem}
The dt-LTI system (11) satisfies ODP($\Sigma$) if and only if there exists a positive-definite matrix $P \succ 0$ such that the following linear matrix inequality holds:
\begin{equation}
M_{1} - \frac{1}{2}(M_{2} + M_{2}^\top) \preceq 0,
\end{equation}
where the corresponding block matrices are defined as:
\begin{equation}
M_{1} := \begin{bmatrix} A^\top P A - P & A^\top P B \\ * & B^\top P B \end{bmatrix} \\,
\end{equation}
\begin{equation}
M_{2} := \begin{bmatrix} -C^\top \Sigma C(A - I) & -C^\top \Sigma CB \\ C(A - I) & CB \end{bmatrix}.
\end{equation}
\end{theorem}

\textit{Proof.} Consider the quadratic storage function $S(x_k) = x_k^\top P x_k$ and define the augmented vector $z_k = [x_k^\top, u_k^\top]^\top$. The variation of the storage function along the discrete-time trajectories is evaluated as:
\begin{align*}
\Delta S_k &= (Ax_k + Bu_k)^\top P (Ax_k + Bu_k) - x_k^\top P x_k \\
&= z_k^\top \begin{bmatrix} A^\top P A - P & A^\top P B \\ B^\top P A & B^\top P B \end{bmatrix} z_k = z_k^\top M_{1} z_k
\end{align*}

Noting that $\Delta y_k = C(A-I)x_k + CBu_k$, substituting it into the supply rate directly yields:
\begin{align*}
u_k^\top \Delta y_k - y_k^\top \Sigma \Delta y_k &= z_k^\top \begin{bmatrix} -C^\top \Sigma C(A - I) & -C^\top \Sigma CB \\ C(A - I) & CB \end{bmatrix} z_k \\
&= z_k^\top M_{2} z_k
\end{align*}

The discrete-time ODP dissipation inequality requires $\Delta S_k \le u_k^\top \Delta y_k - y_k^\top \Sigma \Delta y_k$, which is equivalent to $z_k^\top (M_{1} - M_{2}) z_k \preceq 0$ for any arbitrary state-input pair $z_k$. Enforcing the symmetry of the central matrix yields the LMI condition (15).

\subsection{The Error of Sampling}

Here, we aim to show the relation between the ct-ODP and the dt-ODP. We aim to answer whether we can use the sampled dt-LTI system to infer ODP index of the actual ct-LTI system.

We approximate the matrix exponential by its first order term. It follows from (12) that
\begin{equation}
    A = I + h\bar{A} + O(h^2), \quad B = h\bar{B} + O(h^2), 
\end{equation}
where $O(h^2)$ is the error matrix of order $h^2$. Substituting (18) into (16) and (17), we obtain
\begin{equation}
    M_1 = h\bar{M}_1 + O(h^2), \quad M_2 = h\bar{M}_2 + O(h^2) .
\end{equation}
Notice that $h > 0$. We have
\begin{equation*}
    M_1 - \frac{1}{2}(M_2 + M_2^\top) + O(h^2) \preceq 0 \iff \bar{M}_1 - \frac{1}{2}(\bar{M}_2 + \bar{M}_2^\top) \preceq 0
\end{equation*}
this indicates that the approximation error from sampling is in the order of $h^2$, if we use the LMI (15) to verify ODP($\Sigma$) of the actual ct-LTI system. Note that $O(h^2)$ is generally an indefinite matrix. Hence, the approximation error could be over-estimated or under-estimated. Nevertheless, as $h \to 0$, the approximation will be exact. Practically, we suggest the sampling step size should satisfy $h^2 \ll \lVert M_1 - \frac{1}{2}(M_2 + M_2^\top) \rVert$ to obtain an accurate result.

\section{DATA-DRIVEN ODP CERTIFICATION FROM INPUT-STATE TRAJECTORIES}

This section first establishes the fundamental rank condition for data informativity, then derives a data-driven LMI verification criterion, and formulates an optimization problem to extract the optimal ODP index from measured trajectories.

\subsection{Data Collection and Rank Condition}

Our objective is to verify the discrete-time ODP property without explicitly identifying the system matrices $(A, B)$. Based on the measured trajectories of the system, we collect a sequence of inputs $\{u_k\}_{k=0}^{N-1}$  and corresponding states $\{x_k\}_{k=0}^N$. We then construct the following data matrices:
\begin{align*}
X &:= \begin{bmatrix} x_0 & x_1 & \cdots & x_{N-1} \end{bmatrix}, \\
X_+ &:= \begin{bmatrix} x_1 & x_2 & \cdots & x_N \end{bmatrix}, \\
U &:= \begin{bmatrix} u_0 & u_1 & \cdots & u_{N-1} \end{bmatrix}.
\end{align*}

The proposed data-driven framework relies entirely on the measured input-state trajectories\cite{koch2020verifying}. A fundamental prerequisite for this approach is that the collected data must be sufficiently informative to characterize the underlying system dynamics. This requirement, which plays a pivotal role in the subsequent verification, is strictly formalized by the following full row rank condition\cite{willems2005note}:
\begin{equation}
\text{rank} \begin{bmatrix} X \\ U \end{bmatrix} = n + m.
\end{equation}

This fundamental rank condition can be guaranteed by requiring the input sequence to be sufficiently persistently exciting. Given the finite input sequence $\{u_k\}_{k=0}^{N-1}$, we construct the corresponding Hankel matrix of depth $L$ as follows:
\begin{equation}
H_L(u) := \begin{bmatrix} 
u_0 & u_1 & \cdots & u_{N-L} \\ 
u_1 & u_2 & \cdots & u_{N-L+1} \\ 
\vdots & \vdots & \ddots & \vdots \\ 
u_{L-1} & u_L & \cdots & u_{N-1} 
\end{bmatrix}.
\end{equation}

\begin{definition}[Persistent Excitation\cite{koch2020verifying}]
We say that a sequence $\{u_k\}_{k=0}^{N-1}$ with $u_k \in \mathbb{R}^m$ is persistently exciting of order $L$, if $\text{rank}(H_L(u)) = mL$.
\end{definition}

According to the fundamental lemma, a sufficient condition to guarantee the rank condition in (20) is that the input trajectory $u_k$ is persistently exciting of order $n + 1$ \cite{9551767}.

\subsection{Data-Driven ODP Verification Criterion}

Exploiting discrete-time input-state trajectories completely bypasses explicit physical models. Provided the data satisfies the informativity rank condition, the following theorem establishes a data-driven LMI criterion for ODP verification.

\begin{theorem}
Given input and state trajectories $\{u_k\}_{k=0}^{N-1}$, $\{x_k\}_{k=0}^N$ of the dt-LTI system (11), if the rank condition (20) holds and there exists a $P = P^\top \succ 0$ such that
\end{theorem}
\begin{equation}
X_+^\top P X_+ - X^\top P X - \frac{1}{2}(M_3 + M_3^\top) \preceq 0, 
\end{equation}
where
\begin{equation}
M_3 := [U - \Sigma CX]^\top [CX_+ - CX] ,
\end{equation}
then system (11) is ODP($\Sigma$).

\textit{Proof.} It follows from $X_+ = AX + BU$ that
\begin{equation*}
X_+^\top P X_+ - X^\top P X = \begin{bmatrix} X \\ U \end{bmatrix}^\top \begin{bmatrix} A^\top P A - P & A^\top P B \\ * & B^\top P B \end{bmatrix} \begin{bmatrix} X \\ U \end{bmatrix}
\end{equation*}
and
\begin{equation*}
M_3 = \begin{bmatrix} X \\ U \end{bmatrix}^\top \begin{bmatrix} -C^\top \Sigma C(A - I) & -C^\top \Sigma CB \\ C(A - I) & CB \end{bmatrix} \begin{bmatrix} X \\ U \end{bmatrix}
\end{equation*}
Hence, (22) implies that
\begin{equation}
\begin{bmatrix} X \\ U \end{bmatrix}^\top \left[ M_1 - \frac{1}{2}(M_2 + M_2^\top) \right] \begin{bmatrix} X \\ U \end{bmatrix} \preceq 0 
\end{equation}
By the rank condition (20), LMI (22) is equivalent to (15). It then follows from Theorem 2 that system (11) is ODP($\Sigma$). \hfill $\square$

Our method can also apply to the case that $C$ is unknown but measurements of $\{y_k\}_{k=0}^N$ are available, by simply replacing $CX$ and $CX_+$ with $Y$ and $Y_+$, respectively.

To extract the optimal ODP index matrix $\Sigma$ and maximize the operational margin, we formulate the following convex semi-definite programming (SDP) problem:
\begin{equation}
\begin{aligned}
    \max_{\Sigma,P} \quad & \text{tr}(\Sigma) \\
    \text{s.t.} \quad & X_+^\top P X_+ - X^\top P X - \frac{1}{2}(M_3 + M_3^\top) \preceq 0 \\
    & P \succ 0
\end{aligned} 
\end{equation}
where $\text{tr}(\Sigma)$ denotes the matrix trace. Solving (25) efficiently provides the rigorous theoretical boundary for distributed stability certification.

Solving the SDP yields the optimal matrix $\Sigma^*$. Its equivalent scalar ODP index, extracted as (7), quantifies the local bus's conservative damping margin and is directly applied in Lemma 1 for system-wide stability verification.

\begin{remark} \label{rem:discrete_continuous_equivalence}
While the physical grid operates in continuous time, the proposed certification relies entirely on discrete-time sampled data. As demonstrated in Section III-C, when the sampling step $h$ is sufficiently small, the $O(h^2)$ approximation error becomes negligible.
\end{remark}

\section{CASE STUDY}
In this section, the proposed data-driven ODP characterization and its application to stability assessment are validated through two case studies. Case I identifies the ODP index of a single conventional droop-controlled inverter unit using offline measurements, while Case II extends the analysis to a 3-bus system to demonstrate the stability evaluation.

\subsection{Offline Measurement of a Single CD Inverter}

To validate the proposed method, we consider an offline measurement scenario for a single grid-connected conventional droop-controlled (CD) inverter. The internal dynamics of the CD inverter are governed by the conventional $P-\theta$ and $Q-V$ droop control equations\cite{6987381}:
\begin{equation}
\begin{cases}
    \tau_1 \dot{\theta} = -(\theta - \theta^*) - D_1(P - P^*) \\
    \tau_2 \dot{V} = -(V - V^*) - D_2(Q - Q^*)
\end{cases},
\end{equation}
where $\tau_1, \tau_2$ are the time constants, and $D_1, D_2$ are the droop gains. The specific physical and control parameters utilized in this offline procedure are detailed in Table \ref{tab:cd_params}.

\begin{table}[htbp]
\vspace{-2mm} 
\caption{Parameters of the CD Inverter}
\label{tab:cd_params}
\centering
\renewcommand{\arraystretch}{1.2} 
\setlength{\tabcolsep}{10mm} 
\begin{tabular}{l | l}
\hline
\textbf{Parameters} & \textbf{Values} \\
\hline
$\tau_1, \tau_2$ / s & 1.0, 10.0 \\
$D_1, D_2$ & 0.3, 0.1 \\
$P^*, Q^*, V^*$ / p.u. & 0.5, 0.0, 1.0 \\
$\theta^*$ / rad & 0 \\
\hline
\end{tabular}
\end{table}

To establish a rigorous baseline for verifying the data-driven method, we introduce the analytical ODP bound for the CD inverter. According to the model-based stability conditions derived in \cite{yang2019distributed}, the theoretical maximum ODP index $\sigma_{th}$ is inherently bounded by the droop gains and the steady-state operating point ($V^*, Q^*$):
\begin{equation}
    \sigma_{th} = \min \left\{ \frac{1}{D_1}, \frac{V^* / D_2 + Q^*}{{V^*}^2} \right\} \label{eq:cd_theoretical}.
\end{equation}
Substituting the parameters from Table \ref{tab:cd_params} into \eqref{eq:cd_theoretical} yields the theoretical baseline $\sigma_{th} = 3.33$.

Under $1\%$ Gaussian noise, trajectories are sampled at $h = 0.4$\,s over a 200\,s window with a multi-frequency sinusoidal input. To suppress noise and satisfy the rank condition, the recorded data is averaged across multiple periods and truncated to a 1.6\,s window. 

Given the magnitude disparity between $\tau_1$ and $\tau_2$, a simple data normalization is performed to avoid ill-conditioned matrices. Solving the SDP (25) yields the optimal index matrix $\Sigma_{cd}^*$. By extracting its scalar index, we obtain $\sigma_{cd} = 3.38$. Compared to the theoretical bound of $\sigma_{th} = 3.33$, the data-driven result exhibits a marginal overestimation of $1.42\%$. 

This slight deviation is primarily caused by finite data truncation and measurement noise, which is completely acceptable in practical applications. Achieving this result without explicit physical models verifies the effectiveness of the proposed framework, providing a solid foundation for the system-wide online assessment in Section V-B.

\subsection{Online Assessment of a 3-Bus System}

To validate the proposed data-driven ODP criteria in a network setting, a 3-bus microgrid is established. Bus 1 is connected with a Synchronous Generator. Bus 2 is attached to a quadratic droop-controlled inverter, and Bus 3 is connected with a conventional droop-controlled inverter as a load. The detailed differential equations governing the internal dynamics of these three devices can be found in \cite{yang2019distributed}.

The physical parameters of the transmission lines and the devices are detailed in Table \ref{tab:system_params}.

\begin{table}[htbp]
\caption{System Parameters of the 3-Bus Microgrid}
\label{tab:system_params}
\centering
\renewcommand{\arraystretch}{1.2} 
\setlength{\tabcolsep}{3.5mm} 
\begin{tabular}{l | l}
\hline
\textbf{Parameters} & \textbf{Values (p.u.)} \\
\hline
Transmission Lines $x, r$ & 0.12, 0.01 \\
SG: $M_i, D_i, T'_{di}, x_{di}, x'_{di}$ & 0.16, 0.076, 6.56, 0.295, 0.17 \\
QD parameters $\tau_1, \tau_2$ / s & 0.3, 8.0 \\
CD parameters $\tau_1, \tau_2$ / s & 1.0, 10.0 \\
\hline
\end{tabular}
\end{table}

To simulate a stressed and challenging operating condition, a heavy-load scenario is configured. The base power profile is set with active power $P_2 = 2.0$ at Bus 2, and a load demand of $P_3 = -3, Q_3 = -0.2$ at Bus 3.

To evaluate the baseline stability, the inherent network ODP requirement, $\sigma_{net}$, is first calculated. We then continuously scale down the theoretical control parameter $\sigma$ around the target threshold $\sigma_{net}$ (specifically, varying from $\sigma_{net} + 0.3$ to $\sigma_{net} - 0.2$). During this shift, the internal control gains are updated according to the mapping theoretical conditions to strictly maintain the steady-state equilibrium point. 

For each configured $\sigma$, we collect the transient operational data at a sampling step of $h=1$ms and employ the LMI-based algorithm to compute the optimal matrix $\Sigma_i$ for each device. the absolute physical stability is evaluated by calculating the maximum dynamic real part $\max(\text{Re}(\lambda))$ of the system-wide eigenvalues (Fig. \ref{fig:verification}(A)). Simultaneously, The equivalent scalar ODP index is then extracted as (7) (Fig. 1(B)).

\begin{figure}[htbp]
    \centering
    \hspace*{-1mm}\includegraphics[width=0.8\linewidth]{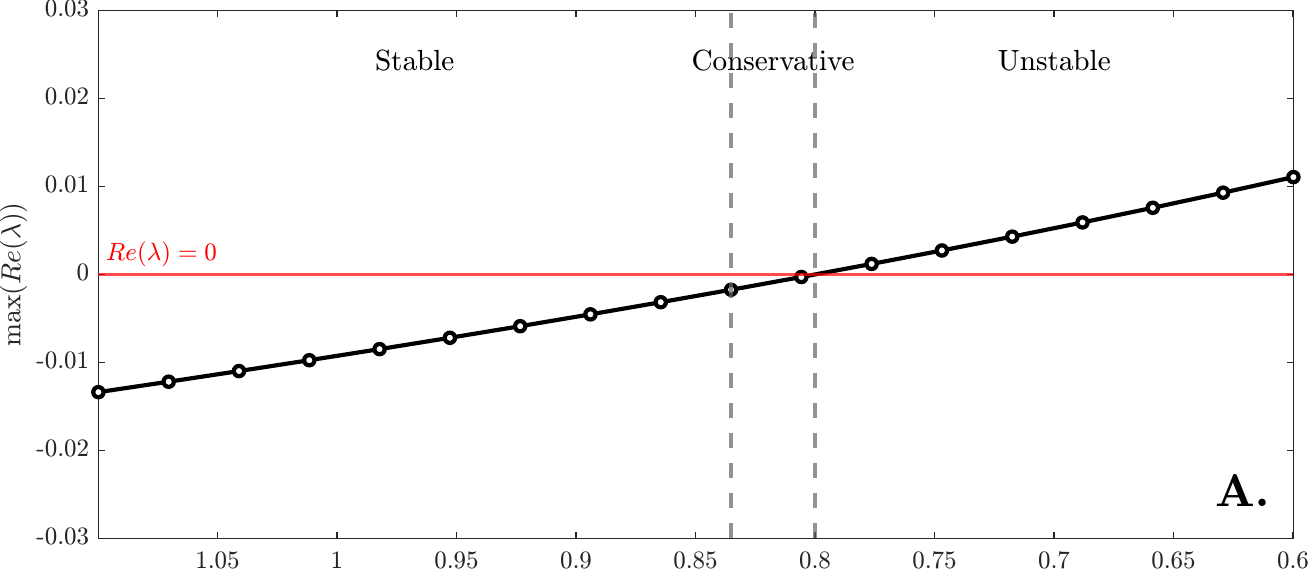}\\[0.4cm]
    \includegraphics[width=0.8\linewidth]{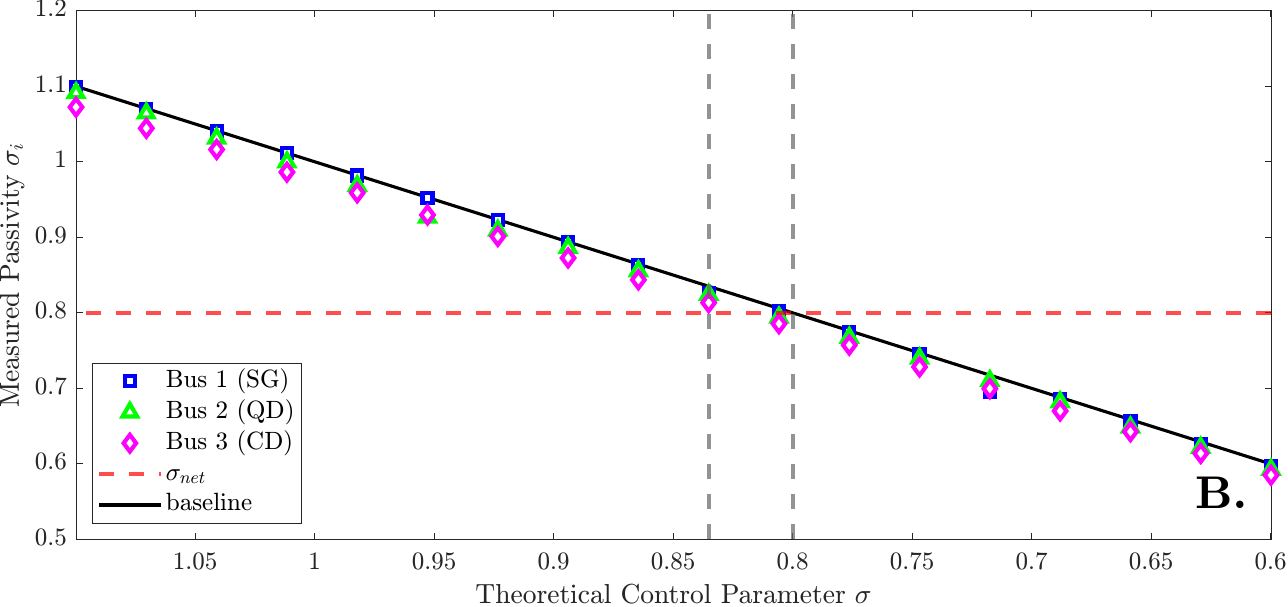}
    \caption{Distributed ODP verification and small-signal stability analysis based on data-driven methods. (A) Trajectory of the maximum dynamic eigenvalue real part. (B) Scatter points represent the data-driven measured ODP indices, and the solid line indicates the theoretical $\sigma$. The two vertical dashed lines bound the narrow conservative region.}
    \label{fig:verification}
\end{figure}

It is clear that the data-driven $\sigma$ values obtained from the devices are in close agreement with the theoretical values. Specifically, when the measured indices satisfy $\min(\sigma_i) > \sigma_{net}$, the system-wide stability is guaranteed. This corresponds to the ``Stable'' region, where all physical eigenvalues possess negative real parts. It should be noted that as $\sigma$ further decreases, the data-driven condition is violated slightly before the physical system actually loses stability. This discrepancy creates a narrow ``Conservative'' region bounded by two vertical dashed lines. As $\sigma$ continues to decrease beyond this region, the physical system eventually becomes unstable. 

Finally, it is important to emphasize that due to the severe ill-conditioning of the data matrices, the proposed method occasionally yields anomalous values, with a failure probability of approximately 6\%.

\section{CONCLUSION}

In this paper, a data-driven distributed framework is proposed to guarantee interconnected power system stability. Bypassing explicit physical models, the output differential passivity of heterogeneous bus dynamics is quantified exclusively using input-state trajectories. Validated through single-inverter offline measurements and a stressed 3-bus online assessment, results confirm that system-wide stability is strictly guaranteed.

Notably, the criteria's theoretical conservativeness ensures absolute safety in stability assessment. Future studies will extend this scalable paradigm to larger networks and improve its numerical conditioning to eliminate anomalous calculations, further enhancing overall algorithmic accuracy.

\bibliographystyle{IEEEtran}
\bibliography{bib/references}

@article{hatziargyriou2020stability,
  title={Stability definitions and characterization of dynamic behavior in systems with high penetration of power electronic interfaced technologies},
  author={Hatziargyriou, Nikos and Milanovi{\'c}, Jovica and Rahmann, Claudia and Ajjarapu, Venkataramana and Canizares, Claudio and Erlich, Istvan and Hill, David and Hiskens, Ian and Kamwa, Innocent and Pal, Bikash and others},
  journal={IEEE PES Technical Report PES-TR77},
  year={2020},
  publisher={IEEE}
}

@inproceedings{milano2018foundations,
  title={Foundations and challenges of low-inertia systems},
  author={Milano, Federico and D{\"o}rfler, Florian and Hug, Gabriela and Hill, David J and Verbi{\v{c}}, Gregor},
  booktitle={2018 power systems computation conference (PSCC)},
  pages={1--25},
  year={2018},
  organization={IEEE}
}

@article{stott2005power,
  title={Power system dynamic response calculations},
  author={Stott, Brian},
  journal={Proceedings of the IEEE},
  volume={67},
  number={2},
  pages={219--241},
  year={2005},
  publisher={IEEE}
}

@article{sastry2003hierarchical,
  title={Hierarchical stability and alert state steering control of interconnected power systems},
  author={Sastry, Shankar and Varaiya, Pravin},
  journal={IEEE Transactions on Circuits and systems},
  volume={27},
  number={11},
  pages={1102--1112},
  year={2003},
  publisher={IEEE}
}

@article{chang1995direct,
  title={Direct stability analysis of electric power systems using energy functions: theory, applications, and perspective},
  author={Chang, Hsiao-Dong and Chu, Chia-Chi and Cauley, Gerry},
  journal={Proceedings of the IEEE},
  volume={83},
  number={11},
  pages={1497--1529},
  year={1995},
  publisher={IEEE}
}

@inproceedings{yang2019toward,
  title={Toward distributed stability analytics for power systems with heterogeneous bus dynamics},
  author={Yang, Peng and Liu, Feng and Wang, Zhaojian and Shen, Chen and Yi, Jun and Lin, Weifang},
  booktitle={2019 IEEE 58th Conference on Decision and Control (CDC)},
  pages={7518--7523},
  year={2019},
  organization={IEEE}
}

@article{song2017distributed,
  title={A distributed framework for stability evaluation and enhancement of inverter-based microgrids},
  author={Song, Yue and Hill, David J and Liu, Tao and Zheng, Yu},
  journal={IEEE Transactions on Smart Grid},
  volume={8},
  number={6},
  pages={3020--3034},
  year={2017},
  publisher={IEEE}
}

@article{yang2019distributed,
  title={Distributed stability conditions for power systems with heterogeneous nonlinear bus dynamics},
  author={Yang, Peng and Liu, Feng and Wang, Zhaojian and Shen, Chen},
  journal={IEEE Transactions on Power Systems},
  volume={35},
  number={3},
  pages={2313--2324},
  year={2019},
  publisher={IEEE}
}

@article{nandanoori2020distributed,
  title={Distributed small-signal stability conditions for inverter-based unbalanced microgrids},
  author={Nandanoori, Sai Pushpak and Kundu, Soumya and Du, Wei and Tuffner, Frank K and Schneider, Kevin P},
  journal={IEEE Transactions on Power Systems},
  volume={35},
  number={5},
  pages={3981--3990},
  year={2020},
  publisher={IEEE}
}

@inproceedings{men2023artificial,
  title={Artificial intelligence (AI) based black-box modeling algorithm for system identification in three-phase single-stage PV inverter systems},
  author={Men, Yuxi and Zhang, Junhui and Lu, Xiaonan and Hong, Tianqi},
  booktitle={2023 IEEE Energy Conversion Congress and Exposition (ECCE)},
  pages={27--32},
  year={2023},
  organization={IEEE}
}

@article{ruiz2019data,
  title={Data-driven control of LVDC network converters: Active load stabilization},
  author={Ruiz-Martinez, Omar F and Mayo-Maldonado, Jonathan C and Escobar, Gerardo and Frias-Araya, Benjamin A and Valdez-Resendiz, Jesus E and Rosas-Caro, Julio C and Rapisarda, Paolo},
  journal={IEEE Transactions on Smart Grid},
  volume={11},
  number={3},
  pages={2182--2194},
  year={2019},
  publisher={IEEE}
}

@article{willems2005note,
  title={A note on persistency of excitation},
  author={Willems, Jan C and Rapisarda, Paolo and Markovsky, Ivan and De Moor, Bart LM},
  journal={Systems \& Control Letters},
  volume={54},
  number={4},
  pages={325--329},
  year={2005},
  publisher={Elsevier}
}

@article{de2019formulas,
  title={Formulas for data-driven control: Stabilization, optimality, and robustness},
  author={De Persis, Claudio and Tesi, Pietro},
  journal={IEEE Transactions on Automatic Control},
  volume={65},
  number={3},
  pages={909--924},
  year={2019},
  publisher={IEEE}
}

@article{zhang2015data,
  title={Data-driven control for interlinked AC/DC microgrids via model-free adaptive control and dual-droop control},
  author={Zhang, Huaguang and Zhou, Jianguo and Sun, Qiuye and Guerrero, Josep M and Ma, Dazhong},
  journal={IEEE Transactions on Smart Grid},
  volume={8},
  number={2},
  pages={557--571},
  year={2015},
  publisher={IEEE}
}

@article{xu2021data,
  title={Data-driven inverter-based Volt/VAr control for partially observable distribution networks},
  author={Xu, Tong and Wu, Wenchuan and Hong, Yiwen and Yu, Junjie and Zhang, Fazhong},
  journal={CSEE Journal of Power and Energy Systems},
  volume={9},
  number={2},
  pages={548--560},
  year={2021},
  publisher={CSEE}
}

@ARTICLE{9551767,
  author={Koch, Anne and Berberich, Julian and Allgöwer, Frank},
  journal={IEEE Transactions on Automatic Control}, 
  title={Provably Robust Verification of Dissipativity Properties from Data}, 
  year={2022},
  volume={67},
  number={8},
  pages={4248-4255},
  doi={10.1109/TAC.2021.3116179}}

@inproceedings{koch2020verifying,
  title={Verifying dissipativity properties from noise-corrupted input-state data},
  author={Koch, Anne and Berberich, Julian and Allg{\"o}wer, Frank},
  booktitle={2020 59th IEEE Conference on Decision and Control (CDC)},
  pages={616--621},
  year={2020},
  organization={IEEE}
}

@ARTICLE{6987381,
  author={Zhang, Yun and Xie, Le},
  journal={IEEE Transactions on Power Systems}, 
  title={Online Dynamic Security Assessment of Microgrid Interconnections in Smart Distribution Systems}, 
  year={2015},
  volume={30},
  number={6},
  pages={3246-3254},
  doi={10.1109/TPWRS.2014.2374876}}

@article{liu2022stability,
  title={Stability and control of power grids},
  author={Liu, Tao and Song, Yue and Zhu, Lipeng and Hill, David J},
  journal={Annual review of control, robotics, and autonomous systems},
  volume={5},
  number={1},
  pages={689--716},
  year={2022},
  publisher={Annual Reviews}
}

@article{guerrero2012advanced,
  title={Advanced control architectures for intelligent microgrids—Part I: Decentralized and hierarchical control},
  author={Guerrero, Josep M and Chandorkar, Mukul and Lee, Tzung-Lin and Loh, Poh Chiang},
  journal={IEEE Transactions on Industrial Electronics},
  volume={60},
  number={4},
  pages={1254--1262},
  year={2012},
  publisher={IEEE}
}
\vspace{12pt}

\end{document}